\documentclass[12pt]{amsart}

\usepackage{url}

\usepackage{amssymb}

\usepackage{enumerate}

\usepackage{graphicx}

\makeatletter
\@namedef{subjclassname@2010}{%
  \textup{2010} Mathematics Subject Classification}
\makeatother

\newtheorem{thm}{Theorem}[section]

\newtheorem{prop}[thm]{Proposition}

\theoremstyle{definition}

\numberwithin{equation}{section}

\begin{document}

\title[A note on the worst case approach]{A note on the worst case approach for a market with  a stochastic interest rate}

\author[D.Zawisza]{Dariusz Zawisza}

\address{\noindent Dariusz Zawisza, \newline \indent Faculty of Mathematics and Computer Science \newline \indent  Jagiellonian University in Krakow \newline \indent{\L}ojasiewicza  6 \newline \indent 30-348 Krak{\'o}w, Poland }

\email{dariusz.zawisza@im.uj.edu.pl}

\date{\today}

\begin{abstract}
We solve robust optimization problem and show the example of the market model for which the worst case measure is not a martingale measure. In our model the instantaneous interest rate is determined by the Hull-White model and the investor employs the HARA utility to measure his satisfaction.To protect against the model uncertainty he uses the worst case measure approach. The problem is formulated as a stochastic game between the investor and the market from the other side. PDE methods are used to find the saddle point and the precise verification argument is provided.  
\end{abstract}



\maketitle

{\tiny Published in Appl. Math. (Warsaw), 45 (2018) 151--160, \url{https://doi.org/10.4064/am2348-2-2018}} 

\section{Introduction}

 We consider a portfolio problem embedded into a game theoretic problem. We assume that the investor does not trust his model much and believes it is only the best guess based on existing data. In such situation we say that the investor faces the model uncertainty (or the model ambiguity). In this work we would like to put more light into the portfolio optimization problem under the assumption that the short term interest rate exhibits some stochastic nature. We consider a financial market consisting of $n$ assets and a bank account. The interest rate on the bank account follows the Hull--White model, which is extended version of the Vasicek model. The investor chooses between holding cash in a bank account and holding  risky assets. The same model has been considered first by Korn and Kraft \cite{Korn} but without the model uncertainty assumption. Instead of supposing that we have the exact model, we assume here the whole family of equivalent models, which will be described later. To determine robust investment controls the investor maximizes the total expected  HARA utility of the final wealth after taking the infimum over all possible models. The robust optimization in the diffusion setting has been popularized especially by A. Schied and his coauthors (e.g. Schied \cite{Schied} and references therein). The model ambiguity in the Vasicek model and its extensions has been considered already by Flor and Larsen \cite{Flor}, Sun et al. \cite{Sun}, Munk and Rubtsov \cite{Munk}, Wang and Li \cite{Wang}. However, their objective function is different, because it includes the expression (along the lines of Maenhout \cite{Maenhout}) which penalize the expected utility for divergence from the reference probability measure.  Our model is in fact their limiting model, when their ambiguity coefficients are passing to $+ \infty$ ($0$ respectively).
In the current paper the problem is formulated as a theoretic stochastic game between the market and the investor and the saddle point of this game is determined, despite of the fact we do not include the  penalizing term into the objective function. Moreover, in addition to aferomentioned papers we provide correct and precise verification reasoning. First, we consider the full problem, without any constraints on the set of uncertainty measures. Further, we investigate what are the properties of the restricted model. To solve the game, we use the  Hamilton-Jacobi-Bellman-Isaacs equation. After several substitutions we are able to solve the equation and use suitable version of the verification theorem to justify the method. Previously the same method has been used by Zawisza \cite{Zawisza1}, \cite{Zawisza2}, but in the model with a deterministic  interest rate and with a different objective function.  The major motivation for considering such model is to provide an example in which results of Oksendal and Sulem \cite{Oksendal1},\cite{Oksendal2} do not hold. In the papers they have considered the jump diffusion model but without assuming the stochastic nature of the interest rate, and have discovered that in that game the investor should always choose to invest only in the bank account and at the same time optimal market strategy is to choose a martingale measure. It is interesting because the martingale measure plays prominent role in derivative pricing. Our paper proves that in our framework the worst case measure is different from the martingale measure.

\section{Model description} 
Let $( \Omega, \mathcal{F}, P)$ be a probability space with
filtration $(\mathcal{F}_{t}, 0 \leq t \leq \mathcal{T})$ (possibly enlarged to satisfy usual assumptions) spanned by $n$-dimensional Brownian motion $(W_{t}=(W^{1}_{t},W^{2}_{t},\ldots W^{n}_{t})^{T} ,\; 0 \leq t \leq\mathcal{T})$. We have the initial measure $P$, but our investor concerns model uncertainty, so the measure should be treated only as a proxy for the real life measure. Further, we will consider a whole class of equivalent measures, which will describe the model uncertainty. Our agent has an access to the market with a bank account $(B_{t}, 0 \leq t \leq \mathcal{T})$ and  risky assets $(S_{t}=(S_{t}^{1},S_{t}^{2},\ldots,S_{t}^{n}), 0 \leq t \leq \mathcal{T})$. Under the measure $P$ the system is given by 
\begin{equation} \label{model}
\begin{cases}
dB_{t} &=r_{t} B_{t} dt,  \\ 
dS_{t} &=diag(S_{t}) [(r_{t}e+ \Sigma_{t} \lambda_{t}^{T})   dt +  \Sigma_{t}  dW_{t}],   \\
dr_{t} &= (b_{t}-\kappa_{t}r_{t}) dt + a_{t} dW_{t}.
\end{cases}
\end{equation}
 We assume that $e=(1,1,\ldots,1)$, coefficients $\kappa_{t}$, 
 $b_{t}$, $\lambda_{t}=(\lambda_{t}^{1},\lambda_{t}^{2},\ldots,\lambda_{t}^{n})$, $a_{t}=(a^{1}_{t},a^{2}_{t},\ldots,a^{n}_{t})$, $\Sigma_{t}=[\sigma^{i,j}_{t}]_{i,j=1\ldots n}$ are continuous deterministic functions, and in addition $\Sigma_{t}$ is invertible. For notational convenience we omit the term $a_{t}\lambda_{t}^{T} dt$ in the dynamics for $r$, and we assume it is already included in $b_{t} dt$ term. The representative example for the  process $(S_{t}, t \in [0,\mathcal{T}])$ is the mixed stock-bond model (e.g. Korn and Kraft \cite[Section 2.2]{Korn}):
\[
 \begin{cases}
dS_{t}^{1} &=(r_{t}+\lambda_{t}^{1} \sigma_{t}^{1,1} + \lambda_{t}^{2} \sigma_{t}^{1,2}) S_{t}^{1}  dt + \sigma_{t}^{1,1} S_{t}^{1}  dW_{t}^{1} + \sigma_{t}^{1,2} S_{t}^{1}  dW_{t}^{2},   \\
dS_{t}^{2} &=(r_{t}+\lambda_{t}^{2} \sigma_{t}^{2,2})S_{t}^{2}dt + \sigma_{t}^{2,2} S_{t}^{2}  dW_{t}^{2}, \\
dr_{t} &= (b_{t} -\kappa r_{t}) dt+ a_{t} dW_{t}^{2}.
\end{cases}
\]
Here $S_{t}^{2}$ is the price of the bond in the Vasicek model with the maturity $\mathcal{T}'>\mathcal{T}$, which means that 
$\sigma^{2,2}_{t}=-\frac{a}{\kappa}(1-e^{-\kappa(\mathcal{T}'-t)})$.
 
The portfolio process evolves according to 
\[
dX_{t}^{\pi}=r_{t} X_{t}^{\pi} dt + \pi_{t} \Sigma_{t} \lambda_{t}^{T} X_{t}^{\pi} dt  +    X_{t}^{\pi} \pi_{t} \Sigma_{t} dW_{t}.
\]
The symbol $\mathcal{A}_{t}$ denotes the class of progressively measurable processes $\pi=(\pi^{1},\pi^{2},\ldots,\pi^{n})$ such that 
\[
\int_{t}^{\mathcal{T}} |\pi_{s}|^{2} ds <+\infty \; \text{a.s.}
\]
To describe the model uncertainty or model ambiguity issues we assume that the probability measure is not precisely
known and the investor considers a whole class of possible measures. We follow the approach of Oksendal and Sulem \cite{Oksendal1} or Schied \cite{Schied} 
in defining the set
\begin{equation} \label{radon}
\mathcal{Q}_{\mathcal{T}}:= \biggl\lbrace Q^{\eta}_{\mathcal{T}} \sim P \; \vert \; \frac{dQ^{\eta}_{\mathcal{T}}}{dP} =
\mathcal{E} \biggl( \int \eta_{t} dW_{t}
\biggr)_{\mathcal{T}}\; , \quad \eta \in \mathcal{M}
\biggr\rbrace ,
\end{equation}
where $\mathcal{E}(\cdot)_{t}$ denotes the Doleans-Dade
exponential and $\mathcal{M}$ denotes the set of all,
progressively measurable  processes $\eta=(\eta^{1},\eta^{2},\ldots,\eta^{n})$,
such that 
\[
\mathbb{E} \left[\frac{dQ^{\eta}_{\mathcal{T}}}{dP}\right]^{2} <+\infty.
\]
In the latter part of the paper we assume that the process $\eta$ takes his values in a fixed compact and convex set $\Gamma$.
It is convenient to use the $Q^{\eta}_{\mathcal{T}}$ dynamics of the stochastic system $(X_{t},r_{t})$ i.e.
\begin{equation} 
\begin{cases}
dX_{t}^{\pi}=r_{t} X_{t}^{\pi}dt + \pi_{t}\Sigma_{t}(\lambda_{t}^{T} +\eta_{t}^{T})X_{t}^{\pi} dt  + \pi_{t} \Sigma_{t} X_{t}^{\pi}dW_{t}^{\eta},   \\
dr_{t} = [(b_{t}-\kappa_{t} r_{t}) + a_{t}\eta_{t}^{T} ]dt + a_{t} dW_{t}^{\eta}.
\end{cases}
\end{equation}

Our investor takes into account the model ambiguity and has worst case preferences  (Gilboa and Schmeidler \cite{Gilboa} ), i.e. his aim is to maximize
\begin{equation} \label{problem}
\mathcal{J}^{\pi,\eta}(x,r,t) = \inf_{\eta \in \mathcal{M}} \mathbb{E}_{x,r,t}^{\eta}  U(X_{\mathcal{T}}^{\pi}).
\end{equation}
The symbol $\mathbb{E}_{x,r,t}^{\eta}$ is used to denote the expected value under the measure $Q^{\eta}_{\mathcal{T}}$ when system starts at $(x,r,t)$. Here we assume that $U(x)=\frac{x^{\gamma}}{\gamma}$ with $0<\gamma<1$. The solution for $\gamma <0$ will be the same but due to the fact that $U$ has negative values, it is needed to use few more restrictions and technicalities to complete the proof.  

Here we are interested not only in the optimal portfolio $\pi ^{*}$, but also in the measure $Q_{\mathcal{T}}^{\eta^{*}}$ for which the infimum is attained. Therefore, we are looking for a saddle point $(\pi^{*}, \eta^{*})$ i.e.
\[
\mathcal{J}^{\pi,\eta^{*}}(x,r,t) \leq \mathcal{J}^{\pi^{*},\eta^{*}}(x,r,t) \leq \mathcal{J}^{\pi^{*},\eta}(x,r,t).
\]

\section{The solution}
To solve the problem we will use the Hamilton-Jacobi-Bellman-Isaacs operator given by
\begin{align}
\mathcal{L}^{\pi,\eta}  V (x,r,t) :=  &V_{t}  + \frac{1}{2}
a^{2}_{t} V_{rr} +  \frac{1}{2}  \pi \Sigma_{t} \Sigma_{t}^{T} \pi^{T} x^{2}  V_{xx} +
 \pi \Sigma_{t} a_{t} x V_{xr} \label{operator} \\ &+  \pi \Sigma_{t}(\lambda_{t}^{T} + \eta^{T})
x V_{x}
 + \eta a^{T}_{t}  V_{r} +    (b_{t}-\kappa_{t}r)
V_{r}  +rxV_{x}. \notag
\end{align}

It should be considered together  with the verification theorem. The reasoning behind its proof is of standard type (see for instance Zawisza \cite[Theorem 3.1]{Zawisza1}). Here we present only short sketch, just to emphasis some minor differences. 
\begin{thm}[Verification Theorem]\label{verification_theorem} Suppose there exists a positive function 
\[
V \in \mathcal{C}^{2,2,1}((0,+\infty)\times\mathbb{R} \times [0,\mathcal{T})) \cap \mathcal{C} ([0,+\infty)\times\mathbb{R} \times [0,\mathcal{T}])
\]
and a Markov control 
\[
(\pi^{*}(x,r,t),\eta^{*}(x,r,t))\in\mathcal{A}_{t}\times\mathcal{M},
\]
such that
\begin{align}
&\mathcal{L}^{\pi^{*}(x,r,t),\eta}V(x,r,t) \geq 0 \label{first:in1}, \\
&\mathcal{L}^{\pi,\eta^{*}(x,r,t)}V(x,r,t) \leq 0  \label{second:in1}, \\
&\mathcal{L}^{\pi^{*}(x,r,t),\eta^{*}(x,r,t)}V(x,r,t) = 0\label{third:eq1}, \\
&V(x,r,T)=\frac{x^{\gamma}}{\gamma} \label{terminal:cond1}
\end{align}
\flushright for all $\eta \in \mathbb{R}$, $\pi \in \mathbb{R}$,  $(x,r,t) \in (0,+\infty)\times\mathbb{R} \times [0,\mathcal{T}) $,
\flushleft
and
\begin{equation} \label{uniform}
 \mathbb{E}_{x,r,t}^{\eta} \left[ \sup_{t \leq s \leq \mathcal{T}} \left|V(X_{s}^{\pi^{*}},r_{s},s)\right|\right] < + \infty 
\end{equation}
\flushright for all  $ (x,r,t) \in  [0,+\infty)\times\mathbb{R} \times [0,\mathcal{T}]$,
$\pi \in \mathcal{A}_{t}$, $\eta \in \mathcal{M}$. \flushleft
\medskip

Then
\[J^{\pi,\eta^{*}}(x,r,t) \leq V(x,r,t) \leq J^{\pi^{*},\eta}(x,r,t)\]
\flushright
for all $\pi \in \mathcal{A}_{t}$, $\eta \in \mathcal{M}$,
\flushleft  and
\[ 
V(x,r,t)= J^{\pi^{*},\eta^{*}}(x,r,t). 
\]
\end{thm}
\begin{proof}
Let us fix first $\pi \in \mathcal{A}_{t}$. Consider $Q^{\eta^{*}}_{\mathcal{T}}$ - dynamics of the system $(X_{t}, r_{t})$ and apply the It\^{o} formula using the function $V$. By using inequality \eqref{second:in1} and taking the expectation from both sides, we obtain
\[
V(x,r,t) \geq \mathbb{E}^{\eta^{*}} V(X_{(\mathcal{T}-\varepsilon)\wedge \tau_{n}}, r_{(\mathcal{T}-\varepsilon)\wedge \tau_{n}}, (\mathcal{T}-\varepsilon)\wedge \tau_{n}),
\]
 where  $(\tau_{n}, n \geq 0)$ is a localizing sequence of stopping times. The function $V$ is positive, thus the Fatou Lemma implies
 \[
V(x,r,t) \geq \mathbb{E}^{\eta^{*}}_{x,r,t} V(X_{\mathcal{T}}^{\pi},r_{\mathcal{T}},\mathcal{T}) =\mathbb{E}^{\eta^{*}}_{x,r,t} U(X_{\mathcal{T}}^{\pi})=J^{\pi,\eta^{*}}(x,r,t).
\]
 To prove the reverse inequality we fix $\eta \in \mathcal{M}$ and consider $Q^{\eta}_{\mathcal{T}}$ - dynamics of the system $(X_{t}, r_{t})$. After applying the It\^{o} rule we get
 \[
V(x,r,t) \leq \mathbb{E}^{\eta}_{x,r,t} V(X_{(\mathcal{T}-\varepsilon)\wedge \tau_{n}}^{\pi^{*}}, r_{(\mathcal{T}-\varepsilon)\wedge \tau_{n}}, (\mathcal{T}-\varepsilon)\wedge \tau_{n})
\]
and the same is true with the equality
\[
V(x,r,t) = \mathbb{E}^{\eta^{*}}_{x,r,t} V(X_{(\mathcal{T}-\varepsilon)\wedge \tau_{n}}^{\pi^{*}}, r_{(\mathcal{T}-\varepsilon)\wedge \tau_{n}}, (\mathcal{T}-\varepsilon)\wedge \tau_{n}).
\]
Property \eqref{uniform} and the dominated convergence theorem finish the proof.
\end{proof}

Following Korn and Kraft \cite{Korn} we predict that conditions \eqref{first:in1} -- \eqref{uniform} are satisfied by the function of the form 
\[V(x,r,t)=\frac{x^{\gamma}}{\gamma} e^{f(t)r+g(t)}, \quad f(\mathcal{T})=0, \; g(\mathcal{T})=0.\]

 Substituting it into \eqref{first:in1}-\eqref{third:eq1} and dividing the expression by $\frac{x^{\gamma}}{\gamma}e^{f(t)r+g(t)}$, we get
 \[H^{(\pi,\eta^{*})}(r,t) \leq H^{(\pi^{*},\eta^{*})} (r,t)=0 \leq H^{(\pi^{*},\eta)} (r,t), \quad \pi, \eta \in \mathbb{R}^{n}.\]
 where 
\begin{multline*}
H^{(\pi,\eta)}(r,t):= f'(t)r+g'(t) + \frac{1}{2}
a^{2}_{t} f^{2}(t) + \frac{1}{2} \gamma (\gamma-1) \pi \Sigma_{t} \Sigma_{t}^{T} \pi^{T}  +  \pi \Sigma_{t} a^{T}_{t} \gamma f(t)  \\ + \gamma \pi 
\Sigma_{t} (\lambda_{t}^{T} + \eta^{T})  + \eta a^{T}_{t} f(t) + (b_{t}-\kappa_{t}r) f(t) + \gamma r .
\end{multline*}
Now, it is possible to determine the saddle point. Suppose first that we already have the saddle point $(\pi^{*},\eta^{*})$. Therefore,
\[H^{(\pi,\eta^{*})}(r,t) \leq H^{(\pi^{*},\eta^{*})}(r,t), \quad \pi,\eta \in \mathbb{R}^{n}\]
 and consequently
\[
\pi^{*}_{t}=\frac{1}{(1-\gamma)} (\lambda_{t} + \eta^{*} + a_{t}f(t))\Sigma_{t}^{-1} .
\]
On the other hand, 
\[H^{(\pi^{*},\eta^{*})}(r,t) \leq H^{(\pi^{*},\eta)}(r,t), \quad \eta \in \mathbb{R}^{n}. \]
We should notice first that $H$ forms a linear function in $\eta$. In that case, the only chance to find $\eta^{*}$ is to delete the expression with $\eta$ i.e.
\[
\gamma \pi^{*} 
\Sigma_{t}  + a_{t} f(t)=0.
\]
This means that 
\[\pi^{*}= - \frac{f(t)}{\gamma} a_{t} \Sigma_{t}^{-1}.
\]
So, we should have
\[
\frac{  f(t)}{(1-\gamma) } a_{t} \Sigma_{t}^{-1} + \frac{\lambda_{t}+\eta^{*}_{t}}{ (1-\gamma)} \Sigma_{t}^{-1}=- \frac{a_{t} f(t)}{ \gamma} \Sigma_{t}^{-1},
\]
which yields
\[
\eta^{*}_{t}= -\lambda_{t} -\frac{f(t)}{\gamma} a_{t}.
\]
Substituting $\pi^{*}$ and $\eta^{*}$ into the equation and using the fact that the expression with $\eta$ is equal to $0$,
 we get
\begin{multline*}
f'(t)r+g'(t) + \frac{1}{2}
|a_{t}|^{2} f^{2}(t) + \frac{1}{2}  |a_{t}|^{2} f^{2}(t)  \frac{(\gamma-1)}{\gamma} -  |a_{t}|^{2}  f(t) \\ - \lambda_{t} a^{T}_{t} f(t) + (b_{t}-\kappa_{t}r) f(t) + \gamma r =0.
\end{multline*}
 Thus,
\[
f'(t) - \kappa_{t} f(t) + \gamma=0,
\] 
 \[
 g'(t)+  \frac{1}{2}
|a_{t}|^{2} f^{2}(t) + \frac{1}{2}  |a_{t}|^{2} f^{2}(t)  \frac{(\gamma-1)}{\gamma} -  |a_{t}|^{2}  f(t) - \lambda_{t} a^{T}_{t} f(t) +  b_{t} f(t) = 0.
 \]
More explicit forms are:
\begin{align*}
f(t)&=\gamma e^{-\int_{t}^{\mathcal{T}} \kappa_{s} ds} \int_{t}^{\mathcal{T}} e^{\int_{k}^{\mathcal{T}} \kappa_{s} ds} dk , \\
g(t)&= \int_{t}^{\mathcal{T}} \left[\frac{1}{2} f^{2}(s)|a_{s}|^{2} + \frac{1}{2}  |a_{s}|^{2} f^{2}(s)  \frac{(\gamma-1)}{\gamma} - |a_{s}|^{2}  f(s) - \lambda_{s} a^{T}_{s} f(s) +  b_{s} f(s)\right]ds.
\end{align*}
We can now summarize our preparatory calculations.
\begin{prop} \label{main}
The pair $(\pi^{*},\eta^{*})$ given by
\[
\pi^{*}_{t}= - \frac{ f(t)}{ \gamma} a_{t} \Sigma^{-1}_{t}, \quad \eta^{*}_{t}=-\lambda_{t} -\frac{ f(t)}{\gamma} a_{t}
\]
is a saddle point for problem \eqref{problem}.
\end{prop}

\begin{proof}
Note that $\pi^{*}_{t}$ and $\Sigma_{t}$  are deterministic functions.
To complete the proof we need only to verify that 
\[
\mathbb{E}_{x,r,t}^{\eta} \left[ \sup_{t \leq s \leq T} \left|V(X_{s}^{\pi^{*}},r_{s},s)\right|\right] < + \infty, \quad \eta \in \mathcal{M}.
\]
We have
\[
\mathbb{E}_{x,r,t}^{\eta} \left[ \sup_{t \leq s \leq \mathcal{T}} \left|V(X_{s}^{\pi^{*}},r_{s},s)\right|\right] =  \mathbb{E}_{x,r,t} \frac{dQ^{\eta}}{dP} \left[ \sup_{t \leq s \leq \mathcal{T}} V(X_{s}^{\pi^{*}},r_{s},s)\right].
\]
By the Cauchy - Schwarz inequality
\begin{multline*}
\mathbb{E}_{x,r,t} \frac{dQ^{\eta}}{dP} \left[ \sup_{t \leq s \leq \mathcal{T}} V(X_{s}^{\pi^{*}},r_{s},s)\right] \\ \leq \left[\mathbb{E} \left[ \frac{dQ^{\eta}}{dP}\right]^{2} \right]^{\frac{1}{2}} \left [ \mathbb{E}_{x,r,t} \left[ \sup_{t \leq s \leq T} V^{2}(X_{s}^{\pi^{*}},r_{s},s)\right] \right]^{\frac{1}{2}}.
\end{multline*}
The explicit formula for the function $V$ leads to
\[
V(X_{s}^{\pi^{*}}, r_s,s)= \frac{1}{\gamma} \left[X_{s}^{\pi^{*}}\right]^{\gamma} e^{f(s) r_{s} + g(s)}.
\]
The portfolio process $X_{t}$ is a solution to the linear equation, so 
\[
X_{s}^{\pi^{*}}=x e^{\int_{t}^{s} [r_{l} + \pi_{l}^{*}\Sigma_{l} \lambda_{l}^{T}  -\frac{1}{2}(\pi_{l}^{*} \Sigma_{l} \Sigma_{l}^{T} \pi_{l}^{T*}  )]dl +  \int_{t}^{s} \pi_{l}^{*}  \Sigma_{l}  dW_{l}}.
\]
Note that the process $\zeta_{s}=e^{\int_{t}^{s}\kappa_{l} dl} r_{s}$  has the dynamics
\[
d\zeta_{s}=e^{ \int_{t}^{s} \kappa_{l} dl} b_{s} ds+ e^{\int_{t}^{s} \kappa_{l} dl} a_{s}  dW_{s}.
\]
We have
\[
r_{s}= e^{-\int_{t}^{s}\kappa_{l} dl} \left[ r+ \int_{t}^{s} b_{l}  dl +\int_{t}^{s} a_{l} dW_{l}  \right].
\]

By the stochastic Fubini theorem, the expression $V^{2} (X_{s}^{\pi^{*}},r_s,s)$ can be rewritten in the form
\[
V^{2}(X_{s}^{\pi^{*}}, r_s,s)=x Z_{s} e^{\beta(s)r_{s} + \xi(s)},
\]
where the process $Z_{s}$ is a square integrable martingale, $\beta$, $\xi$   
are bounded and deterministic functions. 

After repeating the Cauchy - Schwarz inequality once more it is now sufficient to prove that for any bounded deterministic function  $\hat{\beta}$ we have
\begin{equation} \label{property}
\mathbb{E}_{r,t} \sup_{t \leq s \leq T} e^{\hat{\beta}(s) \zeta_{s}} <+\infty.
\end{equation}
Note that the following inequality is true
\[
e^{\hat{\beta}(s) \zeta_{s}} \leq e^{\hat{\beta}_{max} \zeta_{s}}  + e^{\hat{\beta}_{min} \zeta_{s}}, 
\]
where
\[
\hat{\beta}_{max} = \max_{t \leq s \leq T}\hat{\beta}(s), \quad \hat{\beta}_{min} = \min_{t \leq s \leq T} \hat{\beta}(s).
\]
Both processes $e^{\hat{\beta}_{max} \zeta_{s}}$, $e^{\hat{\beta}_{min} \zeta_{s}}$ are solutions to linear equations with bounded coefficients and thus usual Lipschitz and linear growth conditions are satisfied. Property \eqref{property} follows from standard estimates for stochastic differential equations (see Pham \cite[Theorem 1.3.16]{Pham}).
\end{proof}

{\bf Concluding remarks}

To conclude the result we show that the measure $Q^{\eta^{*}}_{\mathcal{T}}$ is not a martingale measure i.e. the process $S_{t} e^{-\int_{0}^{t} r_{s} ds}$ is not a $Q^{\eta^{*}}_{\mathcal{T}}$ - martingale. To see this, it is sufficient to write $Q^{\eta^{*}}_{\mathcal{T}}$ dynamics of $S_{t}$:
\[
dS_{t} =diag(S_{t}) \left[\left[r_{t}e-\frac{ f(t) }{\gamma}\Sigma_{t}  a^{T}_{t}\right]  dt + \Sigma_{t}  dW_{t} \right].
\]
At the end, it is worth to compare the robust investment strategy
\[
\pi^{*}_{t}= \frac{1}{(1-\gamma)}(\lambda_{t} + \eta^{*}_{t} + f(t) a_{t} )\Sigma_{t}^{-1}, \quad \eta^{*}_{t}=-\lambda_{t} -\frac{ f(t)}{\gamma} a_{t}
\]
with the solution to the traditional utility maximization problem
\[
\pi^{*}_{t}=\frac{1}{(1-\gamma)}(\lambda_{t} + f(t) a_{t} )\Sigma_{t}^{-1}.
\]  
It is worth noticing as well that  $\pi^{*}$ can be rewritten as 
\[\pi^{*}_{t}= - \frac{f(t)}{\gamma} a_{t} \Sigma_{t}^{-1}= - e^{-\int_{t}^{\mathcal{T}} \kappa_{s} ds} \int_{t}^{\mathcal{T}} e^{\int_{k}^{\mathcal{T}} \kappa_{s} ds} dk \; \left[ a_{t} \Sigma_{t}^{-1} \right].
\]
and it does not depend on the risk aversion coefficient $\gamma$. The same property is true for $\eta^{*}$.

\section{Model uncertainty with restrictions}

From the  practitioner's point of view, it might be interesting to solve the problem with restrictions imposed on the uncertainty set $\mathcal{M}$. 
In this section  we assume that the  class $\mathcal{M}$ consists of all progressively measurable processes taking values in a compact and convex fixed set $\Gamma \subset \mathbb{R}^{n}$.

We can use the same function $H$ 
\begin{multline*}
H^{(\pi,\eta)}(r,t)= f'(t)r+g'(t) + \frac{1}{2}
a^{2}_{t} f^{2}(t) + \frac{1}{2} \gamma (\gamma-1) \pi \Sigma_{t} \Sigma_{t}^{T} \pi^{T}  +  \pi \Sigma_{t} a^{T}_{t} \gamma f(t)  \\ + \gamma \pi 
\Sigma_{t} (\lambda_{t}^{T} + \eta^{T})  + \eta a^{T}_{t} f(t) + (b_{t}-\kappa_{t}r) f(t) + \gamma r .
\end{multline*}

To find the explicit saddle point for the function $H$, we start with solving the upper Isaacs equation 
\begin{equation} \label{upper}
\min_{\eta \in \Gamma} \max_{\pi \in \mathbb{R}^{n}} H^{(\pi,\eta)}(r,t)=0.
\end{equation}
Furthermore, we use known results on max-min theorems (Fan \cite[Theorem 2]{Fan}) to verify
\[
\min_{\eta \in \Gamma} \max_{\pi \in \mathbb{R}^{n}} H^{(\pi,\eta)}(r,t)= \max_{\pi \in \mathbb{R}^{n}}\min_{\eta \in \Gamma} H^{(\pi,\eta)}(r,t).
\]
We can determine a saddle point candidate $(\pi^{*},\eta^{*})$
by finding a Borel measurable function $\eta^{*}$, such that
\[
 \min_{\eta \in \Gamma} \max_{\pi \in \mathbb{R}}  H^{(\pi,\eta)}(r,t)= \max_{\pi \in \mathbb{R}}   H^{ (\pi,\eta^{*})}(r,t)
\]
and a Borel measurable function $\pi^{*}$, such that
\[
  \min_{\eta \in \Gamma} \max_{\pi \in \mathbb{R}}  H^{(\pi,\eta)}(r,t)=\min_{\eta \in \Gamma} H^{(\pi^{*},\eta)}(r,t).
\]
Because the variable  $\eta$ is separated from $r$, equation \eqref{upper} can be split into two equations (the first one  has already been solved): 
\[
f(t)=\gamma e^{-\int_{t}^{\mathcal{T}} \kappa_{s} ds} \int_{t}^{\mathcal{T}} e^{\int_{k}^{\mathcal{T}} \kappa_{s} ds} dk,
\]
and

\begin{align*}
&g'(t) + \frac{1}{2} |a_{t}|^{2} f^{2}(t)+ b_{t} f(t) \\ & + \min_{\eta \in \Gamma} \left[
-\frac{1}{2} \frac{\gamma}{1-\gamma} | \lambda_{t} + \eta + f(t)a_{t} |^{2}+ \frac{\gamma}{1-\gamma} (\lambda_{t} + \eta +  f(t)a_{t})(\lambda_{t} + \eta)^{T} +  f(t) a_{t} \eta^{T}\right]=0.
\end{align*}

Therefore, to find $\eta^{*}$, it is sufficient to determine any Borel measurable minimizer to the expression
\begin{equation} \label{minimizer}
-\frac{1}{2} \frac{\gamma}{1-\gamma} | \lambda_{t} + \eta +  f(t)a_{t}|^{2}+ \frac{\gamma}{1-\gamma} (\lambda_{t} + \eta +  f(t)a_{t})(\lambda_{t} + \eta)^{T} + f(t)a_{t} \eta^{T} . 
\end{equation}
Now, let $\pi^{*}$ be a Borel measurable maximizer of the function
\[
  \min_{\eta \in \Gamma} H^{(\pi,\eta)}(r,t).
\]
Then, $(\pi^{*},\eta^{*})$ is a saddle point for the function $H^{(\pi,\eta)}(r,t)$. In particular, 
\[
H^{(\pi,\eta^{*})}(r,t) \leq H^{(\pi^{*},\eta^{*})}(r,t), \quad \pi \in \mathbb{R}^{n}.
\]
The unique function $\pi^{*}$ which satisfy the above condition is given by 
 \[
\pi^{*}_{t}= \frac{1}{(1-\gamma)}(\lambda_{t} + \eta^{*}_{t} +  f(t)a_{t})\Sigma_{t}^{-1}.
\]

\begin{prop}
Suppose that $\eta^{*}$ is a minimizer of \eqref{minimizer} and
\[
\pi^{*}_{t}= \frac{1}{(1-\gamma)}(\lambda_{t} + \eta^{*}_{t} +  f(t)a_{t})\Sigma_{t}^{-1}.
\]
Then the pair $(\pi^{*},\eta^{*})$
is a saddle point for  problem \eqref{problem} with the restrictions imposed by the set $\Gamma$.
\end{prop}
 
The proof is omitted because it is the repetition of the steps from the proof of Proposition \ref{main}.

{\bf Acknowledgements:} I would like to express my gratitude to the Referee for helping me to improve the first version of the paper.

\end{document}